\documentclass[12pt]{article}

\title
	{
	Functional models and self-modeling property of minimal Dirac operators on the half-line
	}
	
\author
	{
	M.\,I. Belishev
		\thanks{St.\,Petersburg Department of Steklov Mathematical Institute of the Russian Academy of Sciences, Fontanka 27, St.\,Petersburg, Russia, 191023; belishev@pdmi.ras.ru}
	and 
	S.\,A. Simonov
		\thanks{St.\,Petersburg Department of Steklov Mathematical Institute of Russian Academy of Sciences, Fontanka 27, St.\,Petersburg, Russia, 191023; Alferov Academic University of the Russian Academy of Sciences, Khlopina 8/3, St.\,Petersburg, Russia, 194021; Institute of Mathematics, ITMO University, Kronverksky Prospekt, 49, St.\,Petersburg, Russia, 197101; sergey.a.simonov@gmail.com}
	}

\usepackage{amssymb,amsmath,amsthm}
\usepackage{graphicx,color}
\usepackage{array, amsfonts, mathrsfs}
\usepackage{epstopdf}
\usepackage[colorlinks=true]{hyperref}
\usepackage{float}
\usepackage{enumitem}
\usepackage[dvipsnames]{xcolor}

\newtheorem{Lemma}{Lemma}
\newtheorem{Theorem}{Theorem}

\newtheorem*{Remark*}{Remark}
\newtheorem{Corollary}{Corollary}

\newtheorem*{Definition*}{Definition}
\newtheorem{Definition}{Definition}

\newcommand{\m}{\mathscr}

\newcommand{\f}{\mathfrak}
\newcommand{\Dom}{\text{Dom\,}}
\newcommand{\op}{\operatorname}
\newcommand{\bb}{\mathbb}

\date{}

\begin{document}
\maketitle

\begin{abstract}
We prove that minimal Dirac operators on the half-line are self-modeling, which means that such an operator is determined by its arbitrary unitary copy uniquely up to a transformation (shape equivalence) which changes its potential by a constant factor of modulus one. This result is obtained using the wave functional model of the minimal matrix Schr\"odinger operator on the half-line.
\end{abstract}

\subsection*{About the paper}

$\bullet$\quad
Let $(X,\mu)$ be a measurable space and $\m E$ be a Hilbert space, $\m H=L_2(X,\mu;\m E)$ be the Hilbert space of square summable $\m E$-valued functions on $X$, $\f U(\m E)$ be the group of unitary operators in $\m E$. Every $\theta\in\f U(\m E)$ defines operator $\hat\theta\in\f U(\m H)$ acting as $(\hat\theta y)(x)=\theta(y(x))$, $y\in\m H$. Let $L$ be an operator in $\m H$ with the domain $\op{Dom} L$. For $\theta\in\f U(\m E)$ we denote $L_{\theta}:=\hat\theta L\hat\theta^*$ acting on $\op{Dom} L_{\theta}=\hat\theta\op{Dom} L$.

Let $\f G\subset\f U(\m E)$ be a subgroup. Operator $L'$ in $\m H$ is called {\it shape equivalent} to $L$ if there exists $\theta\in\f G$ such that $L'=L_{\theta}$ and $\op{Dom}L'=\op{Dom}L$.

An operator $\dot L$ in an (abstract) Hilbert space $\dot{\m H}$ which is unitarily equivalent to $L$ is called its (unitary) {\it copy}, whereas $L$ is called (functional) {\it model} of $\dot{L}$. 

Operator $L$ is called {\it self-modeling} if it is determined up to shape equivalence by any of its unitary copies.
In applications this means that some canonical procedure exists which enables one to recover $L$ or, more precisely, to recover the shape equivalence class $\{L_{\theta}\mid\theta\in\f G\}$, from the copy $\dot L$. In inverse problems of mathematical physics various kinds of inverse data, such as Titchmarsh--Weyl function, characteristic function, response operator, spectral function or other spectral data, provide a unitary copy $\dot{L}$ of the operator $L$ which is to be determined. To determine it one can extract the copy $\dot{L}$ from the data and then solve the {\it abstract inverse problem} by constructing a functional model of the copy. The {\it boundary control method} provides such a procedure in many cases \cite{Belishev1997,Belishev2007,Belishev2017,Belishev2013,BelishevSimonov2020,BelishevSimonov2026JOT}.

Self-modeling property is a property of a class of operators on a functional space with a fixed group $\f G$ of unitary transforms on the space of values $\m E$. The group $\f G$ is chosen in such a way that corresponding transformations preserve the class of operators. For example, for Schr\"odinger and for Dirac operators these groups are different, they preserve the main differential terms.

In the present paper we establish self-modeling property of minimal Schr\"o\-din\-ger and Dirac operators on the half-line. For matrix Schr\"odinger operators self-modeling follows directly from our earlier results \cite{BelishevSimonov2026JOT,BelishevSimonov2024} where the wave functional model was constructed. For Dirac operators these results are also utilized to show self-modeling, however, no functional model of Dirac operator is used, because known functional models of Dirac operators \cite{BelishevSimonov2022} are not suitable for solving this abstract inverse problem. This paper is a detailed version of the short note \cite{BelishevSimonov2026ShortNote}.

This work was supported by the Ministry of Science and Higher Education of the Russian Federation (agreement 075-15-2025-344 dated 29/04/2025 for Saint Petersburg Leonard Euler International Mathematical Institute at PDMI RAS).

\subsection*{Minimal Dirac operators on the half-line}

Let $p \in W_{\infty, \text {loc}}^{1}\left(\mathbb{R}_{+} ; \mathbb{C}\right)$ and consider the minimal Dirac operator
$D_{\text {min}}(p)$ which acts in the Hilbert space $\m H=L_{2}\left(\mathbb{R}_{+} ; \mathbb{C}^{2}\right)$ by the differential expression
\begin{equation*}
\left(\m  D y\right)(x):=i \sigma_{3} y^{\prime}(x)+P(x) y(x),
\end{equation*}
$D_{\text {min}}(p) y=\m D y$, where $\sigma_{3}=\begin{pmatrix} 1 & 0 \\ 0 & -1\end{pmatrix}$ and $P(x)=
\begin{pmatrix} 0 & p(x) \\ \overline{p(x)} & 0\end{pmatrix}$, on the domain
\begin{equation*}
\text {Dom\,} D_{\text{min}}(p)=\left\{y \in \m {H} \mid y \in H_{\text {loc}}^{1}\left(\mathbb{R}_+; \mathbb{C}^{2}\right), \m  D y \in \m {H}, y(0)=0\right\}.
\end{equation*}
It is a symmetric completely non-self-adjoint operator with deficiency indices $n_{+}\left(D_{\min}(p)\right)=n_{-}\left(D_{\min}(p)\right)=1$. Its adjoint is the maximal Dirac operator $D_{\text{max}}(p)=D_{\text {min}}^*(p)$, which acts by the same rule $\m D$ on the domain
\begin{equation*}
\text { Dom } D_{\text{max}}(p)=\left\{y \in \m {H} \mid y \in H_{\text{loc}}^{1}\left(\mathbb{R}_{+}; \mathbb{C}^{2}\right), \m D y \in \m {H}\right\}.
\end{equation*}
We consider as well a self-adjoint extension $D(p)$ of $D_{\text{min}}(p)$ which is a restriction of $D_{\rm max}(p)$ to the domain
$$
\Dom D(p)=\left\{y \in \Dom D_{\text {max}}(p) \mid y_{1}(0)+y_{2}(0)=0\right\}.
$$

\noindent$\bullet$\quad
Let $\f G_D:=\left\{\eta_a=\begin{pmatrix} 1 & 0 \\ 0 & e^{-ia}
\end{pmatrix}\bigg| a\in\bb R\right\}\subset\f U(\bb C^2)$. Note that the equality $\theta\sigma_3\theta^*=\sigma_3$ is possible for $\theta\in\f U(\bb C^2)$, only if $\theta=e^{ib}\eta_a$ with $a,b\in\bb R$, in which case
$$
(D_{\text{min}}(p))_{\theta}y=i\sigma_3y'+\begin{pmatrix} 0 & pe^{ia} \\ \overline pe^{-ia} & 0
\end{pmatrix}y
$$
(independently of $b$).

\begin{Definition}\label{def Dirac eq}
We call two minimal Dirac operators $D_{\rm min}(p_1)$ and $D_{\rm min}(p_2)$ \textbf{shape equivalent}, if $D_{\rm min}(p_1)=\hat\eta D_{\rm min}(p_2)\hat\eta^*$ with some $\eta\in\f G_D$. The potentials of two shape equivalent Dirac operators $p_1$ and $p_2$ are related by the equality $p_1(x)=p_2(x)e^{ia}$ with constant $a\in\bb R$. In this context we call $p_1$ and $p_2$ \textbf{equivalent}.
\end{Definition}

Note that $D_{\text{min}}(\overline{p})$ is unitarily equivalent to $-D_{\text{min}}(p)$, since for $\theta=\begin{pmatrix} 0 & -1 \\ 1 & 0 \end{pmatrix}$ one has $\theta\sigma_3\theta^*=-\sigma_3$ and $\theta P\theta^*=-\begin{pmatrix} 0 & \overline{p} \\ p & 0 \end{pmatrix}$, so $\hat\theta D_{\rm min}(p)\hat\theta^*=-D_{\rm min}(\overline p)$.

\begin{Definition}\label{def exceptional case}
Let us call \textbf{exceptional case} a situation that $D_{\rm{min}}(p)$ is unitarily equivalent to $-D_{\rm{min}}(p)$.
\end{Definition}

This coincides with the condition of $D_{\rm min}(p)$ being unitarily equivalent to $D_{\rm min}(\overline{p})$ and happens, e.\,g., if $p(x)$ has constant argument. In what follows it is important that $D_{\text{min}}^2(\overline{p})$ is unitarily equivalent to $D_{\text{min}}^2(p)$. The main result of the paper is the next theorem.

\begin{Theorem}\label{thm Dirac}
Minimal Dirac operators $D_{\rm min}(p)$ with potentials $p\in W^1_{\infty,\rm{loc}}(\bb R_+;\bb C)$ in the non-exceptional case are self-modeling.
\end{Theorem}

This theorem is an immediate consequence of Threorem \ref{thm Dirac real} below. To prove the result we use the wave functional model of the minimal matrix Schr\"odinger operator based on triangular factorization of the control operator \cite{BelishevSimonov2024,BelishevSimonov2026JOT}. This model also yields self-modeling property of that class of operators.

\subsection*{Self-modeling property of minimal Schr\"odinger operators}

In this section we assume $n\in\bb N$ to be arbitrary, not necessarily equal to $2$.

\noindent$\bullet$\quad
Let $Q\in L_{\infty,\text{loc}}(\mathbb R_+;M_{\mathbb C}^n)$ and $Q(x)=Q^{*}(x)$, $x \geqslant 0$. The maximal Schr\"odinger operator $S_{\rm{max}}(Q)$ in $\m H=L_2(\bb R_+;\bb C^n)$ is defined by the differential expression
\begin{equation*}
(\m S y)(x):=-y^{\prime \prime}(x)+Q(x) y(x),
\end{equation*}
$S_{\rm{max}}(Q) y=\m S y$, on the domain
\begin{equation*}
\op{Dom} S_{\rm{max}}(Q)
=\left\{y \in\m H \mid y \in H_{\text{loc}}^{2}(\mathbb{R}_{+} ; \mathbb{C}^{n}), \m S y \in\m H\right\}.
\end{equation*}
Its adjoint is the minimal operator $S_{\rm{min}}(Q)=S_{\rm{max}}^*(Q)$ which acts by the same rule on a smaller domain. The minimal operator is symmetric. We suppose that $Q$ is such that the operator $S_{\rm{min}}(Q)$ is semibounded from below. In this case the matrix version \cite{ClarkGesztesy2003} of the Povzner--Wienholtz theorem says that $S_{\text {min}}(Q)$ has deficiency indices $(n,n)$, and its domain is
\begin{equation}\label{dom s min}
\op{Dom} S_{\rm{min}}(Q)
=\left\{y \in\op{Dom} S_{\rm{max}}(Q) \mid y(0)=y^{\prime}(0)=0\right\}
\end{equation}
(there are no boundary conditions at infinity).

For $\theta\in\f U(n)$ the operator $S_{\min}(Q_{\theta})$ with the potential 
$Q_{\theta}(x):=\theta Q(x) \theta^{*}$  coincides with the operator $(S_{\rm{min}}(Q))_{\theta}=\hat \theta S_{\rm min}(Q)\hat\theta^*$ and thus is unitarily equivalent to $S_{\rm min}(Q)$.

\begin{Definition}\label{def Schr eq}
Two minimal Schr\"odinger operators, $S_{\rm min}(Q_1)$ and $S_{\rm min}(Q_2)$, are called \textbf{shape equivalent}, if $S_{\rm min}(Q_1)=\hat\theta S_{\rm min}(Q_2)\hat\theta^*$ with some $\theta\in\f U(n)$.
\end{Definition}

Self-modeling property of Schr\"odinger operators follows from the results of \cite{BelishevSimonov2024,BelishevSimonov2026JOT} by construction of the wave functional model (its coordinate version) of a copy. Let operator $\dot{S}$ in $\dot{\m H}$ be a unitary copy of $S_{\text{min}}(Q)$, $\dot S=\Psi S_{\rm min}(Q)\Psi^*$ with some $\Psi\in\f U(\m H;\dot{\m H})$. By \cite[Theorem 5.1]{BelishevSimonov2026JOT} $S_{\rm min}(Q)$ satisfies Conditions 1--5 of \cite{BelishevSimonov2026JOT}. These conditions are formulated in a unitary-equivalent form, hence the copy $\dot S$ also satisfies them. Then by \cite[Theorem 4.5]{BelishevSimonov2026JOT} the construction of the wave functional model of $\dot S$ gives the minimal Schr\"odinger operator $\dot S_{\rm mod}=S_{\rm min}(\dot Q_{\rm mod})$ with the potential $\dot Q_{\rm mod}\in L_{\infty,\rm loc}(\bb R_+;\m B(\op{Ker}\dot S^*))$. One further step transforms this operator unitarily into $S_{\rm min}(Q_{\theta})$ with $\theta\in\f U(n)$. This solves the corresponding abstract inverse problem: to find $S_{\rm min}(Q)$ from its unitary copy $\dot S$ up to shape equivalence.  Let us outline this construction assuming that $\dot S$ and $S_{\text{min}}(Q)$ are positive-definite. 

\noindent\textbf{1.} 
The operator $\dot S$ determines the ordinary boundary triple \cite{DerkachMalamud2017} $(\dot{\mathscr{K}},\dot\Gamma_1,\dot\Gamma_2)$ with the boundary space $\dot{\m{K}}=\operatorname{Ker} \dot{S}^{*}$ ($\operatorname{dim} \dot{\m K}=n$) and the boundary operators $\dot{\Gamma}_{1,2}: \operatorname{Dom} \dot{S}^{*} \rightarrow \dot{\mathscr{K}}$, $\dot{\Gamma}_{1}=\mathbb{I}-\dot{S}_{\mathrm{F}}^{-1} \dot{S}^{*}$, $\dot{\Gamma}_{2}=\dot{\Pi} \dot{S}^{*}$, where $\dot{S}_{\mathrm{F}}$ is the Friedrichs extension of $\dot{S}$ and $\dot{\Pi}$ is the projection in $\dot{\mathscr{H}}$ onto $\dot{\mathscr{K}}$. The elements of this {\it Vishik triple} are related by the Green's formula
\begin{equation*}
\left(\dot{S}^{*} y, z\right)-\left(y, \dot{S}^{*} z\right)=\left(\dot{\Gamma}_{1} y, \dot{\Gamma}_{2} z\right)-\left(\dot{\Gamma}_{2} y, \dot{\Gamma}_{1} z\right), \quad y, z \in \operatorname{Dom} \dot{S}^{*}.
\end{equation*}

\noindent\textbf{2.} 
The triple $(\dot{\mathscr{K}},\dot\Gamma_1,\dot\Gamma_2)$ for every $T>0$ determines dynamical system $\dot{\alpha}_{S}^{T}$
\begin{align}
\label{DS1}&u''(t)+\dot{S}^{*} u(t)=0 && \text { in } \dot{\mathscr{H}},\quad 0<t<T,\\
\label{DS2}&u(0)=u'(0)=0 && \text { in } \dot{\mathscr{H}},\\
\label{DS3}&\dot{\Gamma}_{1} u(t)=f(t) && \text { in } \dot{\mathscr{K}},\quad 0 \leqslant t \leqslant T,
\end{align}
where $f \in \dot{\mathscr{F}}^{T}=L_{2}([0, T] ; \dot{\mathscr{K}})$ is {\it boundary control}, $u=u^{f}(t)$ is the solution. With the system one associates \textit{control operator} $\dot{W}^{T}: \dot{\mathscr{F}}^{T} \rightarrow \dot{\mathscr{H}}$, $\dot{W}^{T} f= u^{f}(T)$, which maps $\dot{\mathscr{F}}^{T}$ onto $\operatorname{Ran} \dot W^{T}=: \dot{\mathscr{U}}^{T}$ isomorphically, and \textit{connecting operator} $\dot{C}^{T}: \dot{\mathscr{F}}^{T} \rightarrow \dot{\mathscr{F}}^{T}$, $\dot{C}^{T}=\dot{W}^{T *} \dot{W}^{T}$, which is a positive-definite isomorphism of the form $\dot{C}^{T}=\mathbb{I}+\dot{A}^{T}$ with a compact $\dot{A}^{T}$.

\noindent\textbf{3.} 
The space $\dot{\mathscr{F}}^{T}$ contains the nest of subspaces $\dot{\mathfrak{f}}^{T}:=\left\{\dot{\mathscr{F}}^{T, s}\right\}_{0 \leqslant s \leqslant T}$, 
$$
\dot{\mathscr{F}}^{T, s}:=\Big\{f \in \dot{\mathscr{F}}^{T} \Big| \operatorname{supp} f \subset[T-s, T]\Big\},
$$
formed of delayed controls. The operator $\dot{C}^{T}$ admits the unique triangular factorization $\dot{C}^{T}=\dot{V}^{T *} \dot{V}^{T}$, $\dot{V}^{T} \dot{\mathscr{F}}^{T, s} \subset \dot{\mathscr{F}}^{T, s}$, with respect to the nest $\dot{\mathfrak{f}}^{T}$, where the factor has the form $\dot{V}^{T}=\mathbb{I}+\dot{B}^{T}$ with a compact $\dot{B}^{T}$ \cite{BelishevSimonov2024,GohbergKrein2004}.

\noindent\textbf{4.} 
Let $\dot{\m H}_{\text {mod }}^{T}$ be a copy of $L_{2}([0, T];\dot{\mathscr{K}})$ and let us denote the argument of functions from $\dot{\m H}_{\text {mod }}^{T}$ by $\tau$. Let $\dot Y^{T}$ change $\tau$ to $T-\tau$ and let $\dot Y_0^T$ change $t$ to $T-t$. Consider $\dot W_{\text {mod }}^{T}:=\dot Y^{T} \dot{V}^{T}\dot Y_0^T: \dot{\mathscr{F}}^{T} \rightarrow \dot{\m H}_{\text {mod }}^{T}$. The operator
$$
\dot S_{\rm mod}^{*\,T}:=\dot W_{\rm mod}^T\bigg(-\frac{d^2}{dt^2}\bigg|_{\{f\in H^2([0,T];\dot{\m K})\mid f(T)=f'(T)=0\}}\bigg)(\dot W_{\rm mod}^T)^{-1}
$$
in the model space $\dot{\m H}_{\text {mod }}^{T}$ turns out to have the form $\dot S_{\text {mod }}^{*\,T}=-\frac{d^2}{d\tau^2}+\dot{Q}_{\rm mod}^{T}(\tau)$ with an operator-valued function $\dot{Q}_{\rm mod}^{T}(\tau)=(\dot{Q}_{\rm mod}^T(\tau))^*$ which does not depend on $T$: $\dot{Q}_{\rm mod}^{T}=\dot{Q}_{\rm mod}|_{[0,T]}$ and $\dot{Q}_{\rm mod}\in L_{\infty,\text{loc}}(\mathbb R_+;\m B(\dot{\m K}))$.

\noindent\textbf{5.} 
Considering this construction for all $T>0$ one obtains to the model space $\dot{\m H}_{\rm mod }= L_{2}(\mathbb{R}_{+} ; \dot{\mathscr{K}})$ and the Schr\"odinger operator $\dot S_{\text {mod}}^{*}=-\frac{d^2}{d\tau^2}+\dot{Q}_{\rm mod}(\tau)$ on its maximal domain which is the wave functional model of $\dot S^*$ \cite[Theorem 4.5]{BelishevSimonov2026JOT}. Correspondingly, the minimal Schr\"odinger operator with the same potential $\dot S_{\rm mod}$ is the model of $\dot S$. 

\noindent\textbf{6.} 
Choosing an orthonormal base $\dot k_1,...,\dot k_n$ in $\dot{\mathscr{K}}$ and considering the unitary operators $\dot\mu\in\f U(\bb C^n;\dot K)$, $\dot\mu:v\mapsto\sum_{l=1}^nv_l\dot k_l$ and corresponding $\hat{\dot\mu}\in\f U(\m H;\dot{\m H}_{\rm mod})$, we make another unitary transformation
$$
\dot S_{\rm mod,c}:=\hat{\dot\mu}^*\dot S_{\rm mod}\hat{\dot\mu}=-\frac{d^2}{d\tau^2}+\dot Q_{\rm mod,c}(\tau),
$$
where $\dot Q_{\rm mod,c}(\tau)=\dot\mu^*\dot Q_{\rm mod}(\tau)\dot\mu$. This operator is called a {\it coordinate wave functional model} of $\dot S$. It depends on the choice of the base in $\dot{\m K}$.

\begin{Theorem}\label{thm follows from JOT}
    Let $\dot S$ be a unitary copy of a minimal Schr\"odinger operator $S_{\rm min}(Q)$ with a Hermitian matrix-valued potential $Q\in L_{\infty,\rm loc}(\bb R_+;M_{\bb C}^n)$ such that the operator is semibounded from below. Then coordinate wave functional models of $\dot S$ have the form $S_{\rm min}(Q_{\theta})$ with $\theta\in\f U(n)$.
\end{Theorem}

\begin{proof}
The result follows from the fact proved in \cite{BelishevSimonov2026JOT}: the coordinate wave functional model of a positive-definite $S_{\rm min}(Q)$ with a Hermitian matrix-valued $Q\in L_{\infty,\rm loc}(\bb R_+;M_{\bb C}^n)$ is a minimal Schr\"odinger operator which is shape equivalent to $S_{\rm min}(Q)$. One can trace the construction of the wave functional model for $S_{\rm min}(Q)$, which goes literally in the same way as for $\dot S$ (one should remove dots in notation). In particular, one has: 
\\
$\m K=\Psi^*\dot{\m K}$, 
\\
$\m H_{\rm mod}=\hat\Psi^*\dot{\m H_{\rm mod}}=L_2(\bb R_+;\m K)$ (where $(\hat\Psi y)(\tau)=\Psi(y(\tau))$, $\tau\geqslant0$), 
\\
$\m F^T=(\hat\Psi^T)^*\dot{\m F^T}=L_2([0,T];\m K)$ (where $(\hat\Psi^Ty)(\tau)=\Psi(y(\tau))$, $\tau\in[0,T]$), 
\\
$W^T=\Psi^*\dot W^T\hat\Psi^T$, 
\\
$C^T=(\hat\Psi^T)^*\dot C^T\hat\Psi^T$, 
$V^T=(\hat\Psi^T)^*\dot V^T\hat\Psi^T$, 
$W_{\rm mod}^T=(\hat\Psi^T)^*\dot W_{\rm mod}\hat\Psi^T$, 
\\
$S^{*\,T}_{\rm mod}=(\hat\Psi^T)^*\dot S_{\rm mod}^{*\,T}\hat\Psi^T$, 
$S^{*}_{\rm mod}=(\hat\Psi)^*\dot S_{\rm mod}^{*}\hat\Psi$,
\\
$Q_{\rm mod}(\tau)=\Psi^* \dot Q_{\rm mod}(\tau)\Psi$. 
\\
In \cite{BelishevSimonov2026JOT} it was proved that $Q_{\rm mod}(\tau)=\phi Q(\tau)\phi^*$, where $\phi=\lambda(\lambda^*\lambda)^{-\frac12}$, $\lambda:\bb C^n\to\m K$ is an isomorphism, thus $\phi\in\f U(\bb C^n;\m K)$. Then 
$$
\dot Q_{\rm mod}(\tau)=\Psi Q_{\rm mod}(\tau)\Psi^*=\dot\phi Q(\tau)\dot\phi^*,
$$
where $\dot\phi=\Psi\phi\in\f U(\bb C^n;\dot{\m K})$, and
$$
\dot S_{\rm mod,c}=\hat{\dot\mu}^*\dot S_{\rm mod}\hat{\dot\mu}=-\frac{d^2}{dx^2}+\dot\mu^*\dot\phi Q(x)\dot\phi^*\dot\mu=-\frac{d^2}{dx^2}+\theta Q(x)\theta^*=S_{\rm min}(Q_{\theta})
$$
with the matrix $\theta:=\dot\mu^*\dot\phi\in\f U(\bb C^n;\bb C^n)=\f U(n)$ corresponding to the choice of the base. Thus $\dot S_{\text {mod,c }}$ takes the form which is shape equivalent to the original $S_{\text{min}}(Q)$.

In the case $S_{\rm min}(Q)$ is only semibounded from below one has to add $\gamma I$ with $\gamma>0$ sufficiently large to obtain a positive-definite operator. Then one can perform the construction of the model described above and subtract $\gamma I$ from the result, which will give a model of $\dot S$ of the desired form.
\end{proof}

Since unitary equivalence of two operators means that they share the same copy, we immediately obtain the following consequence.

\begin{Corollary}
Let two Hermitian matrix-valued potentials $Q_1,Q_2\in L_{\infty,\rm loc}(\bb R_+;M_{\bb C}^n)$ be such that minimal Schr\"odinger operators $S_{\rm min}(Q_1)$ and $S_{\rm min}(Q_2)$ are semibounded from below. If $S_{\rm min}(Q_1)$ and $S_{\rm min}(Q_2)$ are unitarily equivalent, then they are shape equivalent.
\end{Corollary}

Finally we can formulate the obtained results in terms of self-modeling property.

\begin{Theorem}\label{thm Schrodinger}
Semibounded minimal Schr\"odinger operators $S_{\rm min}(Q)$ with Hermitian matrix-valued potentials $Q\in L_{\infty,\rm loc}(\bb R_+;M_{\bb C}^n)$ are self-modeling.
\end{Theorem}

\begin{Remark*}
In contrast to symmetric minimal Schr\"odinger operators, self-adjoint Schr\"odinger operators do not possess self-modeling property. For example, for scalar potentials shape equivalence means equality, however, there exist unitarily equivalent operators which have different potentials.
\end{Remark*}

In what follows we always assume $n=2$.


\subsection*{Self-modeling property of minimal Dirac operators}

In this section we prove Theorem \ref{thm Dirac}. It is more convenient to formulate it in a more explicit form.

\begin{Theorem}\label{thm Dirac real}
In the non-exceptional case a minimal Dirac operator $D_{\rm min}(p)$ with a potential $p\in W^1_{\infty,\rm{loc}}(\bb R_+;\bb C)$ can be recovered from its unitary copy $\dot D$ uniquely up to shape equivalence.
\end{Theorem}

\begin{proof}
The domain of the square of $D_{\text{min}}(p)$ can be described in the following way:
\begin{multline*}
\operatorname{Dom}D_{\text{min}}^2(p)
=\{y \in \text{Dom\,} D_{\text{min}}(p) \mid D_{\text{min}}(p)y\in \text{Dom\,}D_{\text{min}}(p)\}
\\
=\{y \in \m H \mid y,(i \sigma_{3} y^{\prime}+Py)\in H_{\text{loc}}^{1}(\mathbb{R}_{+};\mathbb C^{2}),
\\
(i \sigma_{3} y^{\prime}+Py),(-y^{\prime \prime}+(i \sigma_{3} P^{\prime}+ P^{2}) y) \in \m H,
\\
y(0)=0,i \sigma_{3} y^{\prime}(0)+P(0) y(0)=0\},
\end{multline*}
since 
\begin{multline*}
\m D^{2} y=i \sigma_{3}\left(i \sigma_{3} y^{\prime}+P y\right)^{\prime}+P\left(i \sigma_{3} y^{\prime}+P y\right)
\\
=-y^{\prime \prime}+i \underbrace{\left(\sigma_{3} P+P \sigma_{3}\right)}_{=0} y^{\prime}
+\left(i \sigma_{3} P^{\prime}+P^{2}\right) y=-y^{\prime \prime}+\left(i \sigma_{3} P^{\prime}+P^{2}\right) y.
\end{multline*}
These conditions are equivalent to $y \in H_{\text{loc}}^2(\mathbb{R}_{+};\mathbb C^{2}), \m D y, \m D^{2} y \in\m H,
y(0)=y^{\prime}(0)=0$. Then for the potential
\begin{equation}\label{pQ}
Q:=i\sigma_{3} P^{\prime}+P^{2}
=\begin{pmatrix}
|p|^{2} & i p^{\prime} \\
-i \overline{p^{\prime}} & |p|^{2}
\end{pmatrix}
\end{equation}
$D_{\text {min}}^{2}(p)$ is a restriction of the maximal Schr\"ondinger operator $S_{\text {max}}(Q)$ to the domain
\begin{equation*}
\operatorname{Dom} D_{\min}^{2}(p)
=\{y \in \operatorname{Dom} S_{\text{max}}(Q) \mid \m D y \in \m H, y(0)=y^{\prime}(0)=0\}.
\end{equation*}
Analogously one has
\begin{multline*}
\operatorname{Dom} D^{2}(p)=\{y \in \operatorname{Dom} S_{\text {max}}(Q) \mid \m D y \in\m H, y_{1}(0)+y_{2}(0)=0, 
\\
i y_{1}^{\prime}(0)-i y_{2}^{\prime}(0)+\overline{p(0)} y_{1}(0)+p(0) y_{2}(0)=0\}.
\end{multline*}
Since the operator $D^{2}(p)$ is self-adjoint, the inclusion $D^{2}(p) \subset S_{\max}(Q)$ means that $S_{\min}(Q) \subset D^{2}(p)$, and hence for every $y \in \operatorname{Dom} S_{\min}(Q)$ one has $\m D y \in\m H$, therefore $S_{\min}(Q) \subset D_{\min}^{2}(p)$. It follows that operator $S_{\text {min}}(Q)$, as a restriction of $D^{2}(p)$, is non-negative, and by the Povzner--Wienholtz theorem, \cite{ClarkGesztesy2003}, \eqref{dom s min} holds (there are no boundary conditions at infinity), therefore the domains of $S_{\min}(Q)$ and $D_{\min}^{2}(p)$ coincide, and then
\begin{equation*}
D_{\min}^{2}(p)=S_{\min}(Q).
\end{equation*}

Given a copy $\dot{D}$ of $D_{\min}(p)$ with $p\in W_{\infty,\rm loc}^1(\bb R_+;\bb C)$ one can take its square $\dot{S}=\dot{D}^{2}$ which is a copy of $S_{\rm min}(Q)$, a positive minimal Schrödinger operator to which Theorem \ref{thm Schrodinger} is applicable. The construction of a coordinate wave model of $\dot S$ gives the operator $S_{\min}(Q_{\theta})$ with some (unknown) $\theta \in\f U(2)$. Having $Q_{\theta}$, we recover a Dirac operator which is shape equivalent to $D_{\rm min}(p)$.

An arbitrary matrix $\theta \in\f U(2)$ has the form
\begin{multline*}
\theta=\theta_{\alpha\beta\gamma\varphi}=e^{i \alpha}
\begin{pmatrix}
\cos \varphi & \sin \varphi e^{i \gamma} \\
-\sin \varphi e^{i \beta} & \cos \varphi e^{i(\beta+\gamma)}
\end{pmatrix}
\\
=e^{i \alpha}
\left(\begin{array}{cc}
1 & 0 \\
0 & e^{i \beta}
\end{array}\right)\left(\begin{array}{cc}
\cos \varphi & \sin \varphi \\
-\sin \varphi & \cos \varphi
\end{array}\right)
\left(\begin{array}{cc}
1 & 0 \\
0 & e^{i \gamma}
\end{array}\right)
\end{multline*}
with some $\alpha,\beta,\gamma\in\mathbb R$ and unique $\varphi\in[0,\frac{\pi}2]$. This can be seen from the following. If $\theta\in\f U(2)$, then $|\det\theta|=1$ and matrix $\tilde\theta:=\theta e^{-i\frac{\arg\det\theta}{2}}$ has $\det\tilde\theta=1$ and is also unitary. Let
\begin{equation*}
\tilde\theta=:
\begin{pmatrix}
a & b \\ c & d
\end{pmatrix},
\quad ad-bc=1.
\end{equation*}
From the condition $\tilde\theta^*=\tilde\theta^{-1}$ one has
$
\begin{pmatrix}
\overline{a} & \overline{c} \\ \overline{b} & \overline{d}
\end{pmatrix}
=
\begin{pmatrix}
d & -b \\ -c & a
\end{pmatrix}, 
$
which is equivalent to $d=\overline{a}$, $c=-\overline{b}$. Then the equality $ad-bc=1$ means that $|a|^2+|b|^2=1$ and $|a|=\cos\varphi$, $|b|=\sin\varphi$ with some $\varphi\in[0,\frac{\pi}2]$. Then 
$\tilde\theta=\begin{pmatrix} \cos\varphi e^{i\alpha} & \sin\varphi e^{i\beta} \\ -\sin\varphi e^{-i\beta} & \cos\varphi e^{-i\alpha} \end{pmatrix}$
with $\alpha,\beta\in[0,2\pi)$. One can rewrite this as
$$
\tilde\theta=e^{i\alpha}
\begin{pmatrix}
1 & 0 \\ 0 & e^{-i(\beta+\alpha)}
\end{pmatrix}
\begin{pmatrix}
\cos\varphi & \sin\varphi \\ -\sin\varphi & \cos\varphi
\end{pmatrix}
\begin{pmatrix}
1 & 0 \\ 0 & e^{i(\beta-\alpha)}
\end{pmatrix}.
$$
from which the assertion about $\theta$ is follows.

Since
\begin{equation*}
Q(x)=|p(x)|^2I+\begin{pmatrix}
0 & ip'(x) \\ \overline{ip'(x)} & 0
\end{pmatrix},
\end{equation*}
we have
\begin{equation*}
Q_{\theta}(x)=\theta Q(x) \theta^{-1}
\\
=|p(x)|^2I+\begin{pmatrix}
r(x) & z(x) \\ \overline{z(x)} & -r(x)
\end{pmatrix},
\end{equation*}
where
\begin{eqnarray*}
& r(x):=(ip'(x)e^{-i\gamma}+\overline{ip'(x)} e^{i \gamma})\sin \varphi \cos \varphi \in\mathbb R,
\\
& z(x):=ip'(x)e^{-i(\gamma+\beta)} \cos^{2} \varphi 
-\overline{ip'(x)}e^{i(\gamma-\beta)} \sin^{2} \varphi\in\mathbb C.
\end{eqnarray*}
If one takes $\theta_1=\theta_{\alpha_1\beta_1\gamma_1\varphi_1}\in \f U(2)$, then $\theta_2:=\theta_1\theta\in \f U(2)$ and hence $\theta_2=\theta_{\alpha_2\beta_2\gamma_2\varphi_2}$ with some $\alpha_2,\beta_2,\gamma_2\in\mathbb R$, $\varphi_2\in[0,\frac{\pi}2]$. Then
\begin{equation*}
\theta_1Q_{\theta}(x)\theta_1^*=\theta_2Q(x)\theta_2^*=Q_{\theta_2}(x)=|p(x)|^2I+\begin{pmatrix}
r_2(x) & z_2(x) \\ \overline{z_2(x)} & -r_2(x)
\end{pmatrix}
\end{equation*}
with
\begin{eqnarray*}
& r_2(x)=(ip'(x)e^{-i\gamma_2}+\overline{ip'(x)} e^{i \gamma_2})\sin \varphi_2 \cos \varphi_2 \in\mathbb R,
\\
& z_2(x)=ip'(x)e^{-i(\gamma_2+\beta_2)} \cos^{2} \varphi_2 
-\overline{ip'(x)}e^{i(\gamma_2-\beta_2)} \sin^{2} \varphi_2\in\mathbb C.
\end{eqnarray*}
Let us find $\theta_1$ such that $r_2(x)\equiv 0$.

\begin{Lemma}\label{lemma r_2}
If $r_2(x)\equiv0$, then either $z_2(x)\equiv p'(x)e^{i\delta}$ or $z_2(x)=\overline{p'(x)}e^{i\delta}$ with some $\delta\in\bb R$.
\end{Lemma}

\begin{proof}
$r_2(x)\equiv0$ can happen in two cases: $\sin\varphi_2\cos\varphi_2=0$ or $ip'(x)e^{-i\gamma_2}+\overline{ip'(x)}e^{i\gamma_2}\equiv 0$. 
In the first case $\varphi_2=0$ or $\varphi_2=\frac{\pi}2$, and  one has $z_2(x)=ip'(x)e^{-i(\gamma_2+\beta_2)}$ or $z_2(x)=-\overline{ip'(x)}e^{i(\gamma_2-\beta_2)}$.
In the second case $\overline{p'(x)}\equiv p(x)e^{-2i\gamma_2}$, and $z_2(x)=ip'(x)e^{-i(\gamma_2+\beta_2)}\cos^2\varphi_2+ip'(x)e^{i(\gamma_2-\beta_2-2\gamma_2)}\sin^2\varphi_2=ip'(x)e^{-i(\gamma_2+\beta_2)}$.
\end{proof}

Calculation yields:
\begin{equation*}
\theta_1Q_{\theta}(x)\theta_1^*
=|p(x)|^2I+\theta_1\begin{pmatrix}
r(x) & z(x) \\ \overline{z(x)} & -r(x)
\end{pmatrix}\theta_1^*
=|p(x)|^2I+\begin{pmatrix}
r_2(x) & z_2(x) \\ \overline{z_2(x)} & -r_2(x)
\end{pmatrix},
\end{equation*}
and
\begin{multline*}
r_2(x):=\cos(2\varphi_1)r(x)+\sin(2\varphi_1)\cos\gamma_1\op{Re}z(x)+\sin(2\varphi_1)\sin\gamma_2\op{Im}z(x),
\end{multline*}
\begin{equation}\label{z2}
z_2(x):=-\sin(2\varphi_1)r(x)+\cos^2\varphi_1e^{-i\gamma_1}z(x)-\sin^2\varphi_1e^{i\gamma_1}\overline{z(x)}.
\end{equation}
The equality
\begin{equation*}
\cos(2\varphi_1)r(x)+\sin(2\varphi_1)\cos\gamma_1\op{Re}z(x)+\sin(2\varphi_1)\sin\gamma_2\op{Im}z(x)\equiv 0
\end{equation*}
means that three real-valued functions $r,\op{Re}z,\op{Im}z$, which are known from $Q_{\theta}$, are linearly dependent. One can find coefficients in their non-trivial zero linear combination and then obtain $z_2$ by the following procedure.
\\
\textbf{1.} 
If $r(x)\equiv\op{Re}z(x)\equiv\op{Im}z(x)\equiv 0$, then trivially $z_2(x)\equiv 0$ and nothing else is needed.
\\
\textbf{2.} 
If there exists $x_1\geqslant0$ such that $(r(x_1),\op{Re}(x_1),\op{Im}(x_1))\neq 0$, while for every $x\geqslant0$ one has $(r(x),\op{Re}(x),\op{Im}(x))\in\op{span}\{(r(x_1),\op{Re}(x_1),\op{Im}(x_1))\}$ (in $\mathbb R^3$), then one can take non-zero $(c_1,c_2,c_3)\perp(r(x_1),\op{Re}(x_1),\op{Im}(x_1))$ to have $c_1r(x)+c_2\op{Re}z(x)+c_3\op{Im}z(x)\equiv 0$.
\\
\textbf{3.} 
If there exist $x_1,x_2\geqslant0$ such that $(r(x_1),\op{Re}(x_1),\op{Im}(x_1))\neq 0$ and 
$$
(r(x_2),\op{Re}(x_2),\op{Im}(x_2))\notin\op{span}\{(r(x_1),\op{Re}(x_1),\op{Im}(x_1))\},
$$
then for every $x\geqslant0$ one should have 
$$
(r(x),\op{Re}(x),\op{Im}(x))\in\op{span}\{(r(x_1),\op{Re}(x_1),\op{Im}(x_1)),(r(x_2),\op{Re}(x_2),\op{Im}(x_2))\}
$$
(otherwise $r$, $\op{Re}z$, $\op{Im}z$ are linearly independent functions). Then one can take non-zero 
$$
(c_1,c_2,c_3)\perp\op{span}\{(r(x_1),\op{Re}(x_1),\op{Im}(x_1)),(r(x_2),\op{Re}(x_2),\op{Im}(x_2))\}
$$
and have $c_1r(x)+c_2\op{Re}z(x)+c_3\op{Im}z(x)\equiv 0$. 

As a result, the coefficients $c_1,c_2,c_3$ can always be found. Then the spherical angles of the point $(c_1,c_2,c_3)\in\mathbb R^3$ give the desired $\gamma_1\in[0,2\pi)$ and $2\varphi_1\in[0,\pi]$, which in turn give $z_2$ by \eqref{z2}.

Eventually one knows $|p(x)|$ and $z_2(x)$ such that $z_2(x)\equiv p'(x)e^{i\delta}$ or $z_2(x)\equiv\overline{p'(x)}e^{i\delta}$ with $\delta\in\bb R$, by Lemma \ref{lemma r_2}. Denote $\hat p(x):=p(x)e^{i\delta}$ in the first case and $\hat p(x):=\overline{p(x)}e^{i\delta}$ in the second. We need to find $\hat p$ such that
\begin{equation}\label{hat p}
\begin{array}{rcl}
|\hat p|&=&|p|,
\\
\hat p'&=& z_2,
\end{array}
\end{equation}
and to determine, whether $\hat p$ is equivalent to $p$ or to $\overline{p}$ (note that $p$ cannot be equivalent to $\overline{p}$, because we consider a non-exceptional case). Consider different possibilities.
\\
\textbf{1.} 
If $|p(0)|=0$, then $\hat p(x)=\int_0^xz_2(t)dt$.
\\
\textbf{2.} 
If $|p(0)|\neq 0$ and $z_2\equiv 0$, then $\hat p(x)=|p(0)|e^{ia}$ with some $a\in\mathbb R$.
\\
\textbf{3.} 
If $|p(0)|\neq 0$ and there exists $x_0\geqslant0$ such that $\int_0^{x_0}z_2(t)dt\neq 0$, then let $\varkappa_0:=\arg\hat p(0)$ and $\varkappa(x):=\arg\int_0^xz_2(t)dt$. One has:
\begin{multline*}
|p(x_0)|^2=|\hat p(x_0)|^2=\left||p(0)|e^{i\varkappa_0}+\int_0^{x_0}z_2(t)dt\right|^2
\\
=|p(0)|^2+\left|\int_0^{x_0}z_2(t)dt\right|^2+2|p(0)|\op{Re}\left(e^{-i\varkappa_0}\int_0^{x_0}z_2(t)dt\right),
\end{multline*}
then
\begin{multline*}
\cos\varkappa_0\op{Re}\int_0^{x_0}z_2(t)dt+\sin\varkappa_0\op{Im}\int_0^{x_0}z_2(t)dt
\\
=\frac{|p(x_0)|^2-|p(0)|^2-\left|\int_0^{x_0}z_2(t)dt\right|^2}{2|p(0)|}
\end{multline*}
and
\begin{equation*}
\cos(\varkappa_0-\varkappa(x_0))
=
\frac{|p(x_0)|^2-|p(0)|^2-\left|\int_0^{x_0}z_2(t)dt\right|^2}{2|p(0)|\left|\int_0^{x_0}z_2(t)dt\right|}.
\end{equation*}
\\
\textbf{3a.}
If there exist $x_1,x_2\geqslant0$ such that $\varkappa(x_1)\neq\varkappa(x_2)(\op{mod}\pi)$ and $\int_0^{x_1}z_2(t)dt$, $\int_0^{x_2}z_2(t)dt\neq0$, then $\varkappa_0$ can be determined from the system
\begin{equation*}
\cos(\varkappa_0-\varkappa(x_1))
=
\frac{|p(x_1)|^2-|p(0)|^2-\left|\int_0^{x_1}z_2(t)dt\right|^2}{2|p(0)|\left|\int_0^{x_1}z_2(t)dt\right|},
\end{equation*}
\begin{equation*}
\cos(\varkappa_0-\varkappa(x_2))
=
\frac{|p(x_2)|^2-|p(0)|^2-\left|\int_0^{x_2}z_2(t)dt\right|^2}{2|p(0)|\left|\int_0^{x_2}z_2(t)dt\right|}.
\end{equation*}
which is a non-degenerate linear system on $\cos\varkappa_0$, $\sin\varkappa_0$.
\\
\textbf{3b.}
If such a pair of points does not exist, then $\arg\int_0^{x}z_2(t)dt\equiv 0$ (since it is continuous, it should be constant and hence zero), which means that $z_2$ is real-valued and
$$
\varkappa_0\in
\left\{\pm\arccos\left(\frac{|p(x_0)|^2-|p(0)|^2-\left|\int_0^{x_0}z_2(t)dt\right|^2}{2|p(0)|\left|\int_0^{x_0}z_2(t)dt\right|}\right)\right\}.
$$ 
Both choices of the sign lead to solutions of the system \eqref{hat p}
$$
\hat p(x)=|p(0)|e^{i\varkappa_0}+\int_0^{x_0}z_2(t)dt,\quad\overline{\hat p(x)}=|p(0)|e^{-i\varkappa_0}+\int_0^{x_0}z_2(t)dt.
$$ 
Since $z_2$ is real-valued, changing the sign of $\varkappa_0$ results in complex conjugation of $\hat p$ which gives another solution. Thus we have found the unique solution of the system \eqref{hat p} in the cases \textbf{1}, \textbf{2}, \textbf{3a} and both its solutions, which are complex conjugate, in the case \textbf{3b}. This means that we have recovered the minimal Dirac operator $D_{\rm min}(p)$ up to shape equivalence, because one (and only one) of the operators $D_{\rm min}(\hat p)$ and $D_{\rm min}(\overline{\hat p})$ is unitarily equivalent to $\dot D$, and that operator should also be shape equivalent to $D_{\rm min}(p)$.
\end{proof}

We note an immediate consequence of the obtained result which is similar to the Schr\"odinger case.

\begin{Corollary}
Let $p_1,p_2\in W^1_{\infty,\rm loc}(\bb R_+;\bb C)$. If two minimal Dirac operators $D_{\rm min}(p_1)$ and $D_{\rm min}(p_2)$ are unitarily equivalent and the case is not exceptional, then they are shape equivalent.
\end{Corollary}


\begin{thebibliography}{10}

\bibitem{Belishev1997}
M.~I. Belishev.
\newblock {Boundary control in reconstruction of manifolds and metrics (the BC
  method)}.
\newblock {\em Inverse Problems}, 13(5):1--45, 1997.

\bibitem{Belishev2007}
M.~I. Belishev.
\newblock {Recent progress in the boundary control method}.
\newblock {\em Inverse Problems}, 23(5):1--67, 2007.

\bibitem{Belishev2013}
M.~I. Belishev.
\newblock {A unitary invariant of a semi-bounded operator in reconstructin of
  manifolds}.
\newblock {\em J. Operator Theory}, 69(2):299--326, 2013.

\bibitem{Belishev2017}
M.~I. Belishev.
\newblock Boundary control and tomography of {Riemannian} manifolds (the
  {BC}-method).
\newblock {\em Russian Math. Surveys}, 72(4):581--644, 2017.

\bibitem{BelishevSimonov2026JOT}
M.~I. Belishev and S.~A. Simonov.
\newblock A model and characterization of a class of symmetric semibounded
  operators.
\newblock {\em To appear in J. Operator Theory}.

\bibitem{BelishevSimonov2026ShortNote}
M.~I. Belishev and S.~A. Simonov.
\newblock Self-modeling property of the one-dimensional {Dirac} operator.
\newblock {\em In preparation}.

\bibitem{BelishevSimonov2020}
M.~I. Belishev and S.~A. Simonov.
\newblock The wave model of a metric space with measure and an application.
\newblock {\em Sb. Math.}, 211(4):521--538, 2020.

\bibitem{BelishevSimonov2022}
M.~I. Belishev and S.~A. Simonov.
\newblock A canonical model of the one-dimensional dynamical {Dirac} system
  with boundary control.
\newblock {\em Evolution Equations and Control Theory}, 11(1):283--300, 2022.

\bibitem{BelishevSimonov2024}
M.~I. Belishev and S.~A. Simonov.
\newblock Triangular factorization and functional models of operators and
  systems.
\newblock {\em Algebra i Analiz}, 36(5):1--26, 2024.

\bibitem{ClarkGesztesy2003}
S.~Clark and F.~Gesztesy.
\newblock On {Povzner--Wienholtz}-type self-adjointness results for
  matrix-valued {Sturm--Liouville} operators.
\newblock {\em Proc. Roy. Soc. Edinburgh Sect. A}, 133(4):747--758, 2003.

\bibitem{DerkachMalamud2017}
V.~F. Derkach and M.~M. Malamud.
\newblock {\em Extension theory of symmetric operators and boudnary problems}.
\newblock Institute of Mathematics, Ukrainian National Academy of Sciences,
  Kiev, 2017.

\bibitem{GohbergKrein2004}
I.~Gohberg and M.~Krein.
\newblock {\em Theory and applications of Volterra operators in Hilbert space}.
\newblock American Mathematical Society, 2004.

\end{thebibliography}

\bigskip
\noindent{\bf Keywords:}\,\,\,triangular factorization of operators, nest theory, functional models.

\smallskip
\noindent{\bf MSC:}\,\,\,47A99,\,\, 47B25,\,\, 35R30.
\\
\noindent{\bf UDK:}\,\,\,517.

\end{document}